\let\arxivprimitiveyear\year
\documentclass{ieeeaccess}
\let\accessyearmacro\year
\let\year\arxivprimitiveyear

\definecolor{accessblue}{cmyk}{1,0.27,0,0.06}

\usepackage[T1]{fontenc}
\usepackage[utf8]{inputenc}
\usepackage[english]{babel}

\usepackage{amsmath,amssymb,amsthm}
\usepackage{mathtools}
\usepackage{graphicx}
\usepackage{booktabs}
\usepackage{array}
\usepackage{multirow}
\usepackage{xcolor}
\usepackage{pifont}
\usepackage{tcolorbox}
\tcbuselibrary{listings,breakable,skins}
\usepackage{listings}
\usepackage{tikz}
\usetikzlibrary{arrows.meta,positioning,fit,shapes.multipart,calc}
\usepackage[hidelinks,
            pdftitle={Lineage-Aware Memory Governance: A Derivation-Gated
                      Framework for Privacy-Preserving Column-Level Access
                      Control in Enterprise AI Agents},
            pdfauthor={Venkata Sangaraju; Sudhir Vissa},
            pdfkeywords={lineage-aware memory governance,enterprise AI agents,
                         privacy-preserving retrieval,column-level access control,
                         multi-agent systems,data provenance,metric governance}]{hyperref}
\usepackage{cleveref}
\usepackage{xurl}
\usepackage{enumitem}
\usepackage{caption}
\usepackage{balance}
\usepackage{placeins}

\newtheorem{theorem}{Theorem}
\newtheorem{corollary}[theorem]{Corollary}
\newtheorem{definition}{Definition}
\newtheorem{assumption}{Assumption}
\newtheorem*{remark}{Remark}

\definecolor{algbg}{HTML}{F8FAFC}
\definecolor{algframe}{HTML}{1E3A5F}
\definecolor{algkw}{HTML}{1D4ED8}
\definecolor{algcmt}{HTML}{6B7280}
\definecolor{algnum}{HTML}{374151}

\lstdefinelanguage{pseudocode}{
  morekeywords={for,if,else,return,true,false,and,or,not,in,to,do,
                begin,end,each,all,break,continue,of,with,not\_subset},
  keywordstyle=\bfseries\color{algkw},
  morecomment=[l]{\%\ },
  commentstyle=\itshape\color{algcmt},
  basicstyle=\ttfamily\tiny,
  lineskip=-1.4pt,
  aboveskip=0pt,
  belowskip=0pt,
  numbers=left,
  numberstyle=\tiny\color{algnum},
  numbersep=6pt,
  xleftmargin=14pt,
  framexleftmargin=14pt,
  breaklines=true,
  showstringspaces=false,
  tabsize=2,
}

\newtcblisting[auto counter]{algorithmbox}[3][]{%
  title={\textbf{Algorithm~\thetcbcounter}: #2\hfill\textit{\small#3}},
  colback=algbg, colframe=algframe, fonttitle=\scriptsize\bfseries,
  listing only, listing options={language=pseudocode},
  boxsep=1pt, top=1pt, bottom=1pt, left=1pt, right=1pt,
  before skip=3pt, after skip=3pt,
  breakable, enhanced, #1
}

\definecolor{myblue}{HTML}{1D4ED8}

\let\year\accessyearmacro

\begin{document}

\title{Lineage-Aware Memory Governance: A Derivation-Gated Framework for
Privacy-Preserving Column-Level Access Control in Enterprise AI Agents}

\author{\authorrefmark{1}Venkata Sangaraju and \authorrefmark{2}Sudhir Vissa}

\address[1]{Independent Researcher (ORCID: 0009-0001-7716-1342)}
\address[2]{SAGE7 AI, Georgetown, Texas, USA (ORCID: 0009-0003-7865-0863)}

\corresp{Corresponding author: Venkata Sangaraju.}

\tfootnote{This research received no specific grant from any funding
agency in the public, commercial, or not-for-profit sectors.}

\history{Accepted for publication in \textit{IEEE Access}. This is the
author's accepted manuscript; the final published version is available at
\url{https://doi.org/10.1109/ACCESS.2026.3730363}.}

\doi{Digital Object Identifier 10.1109/ACCESS.2026.3730363}

\markboth{Sangaraju \MakeLowercase{\textit{et al.}}: Lineage-Aware Memory
Governance for Enterprise AI Agents}{Sangaraju \MakeLowercase{\textit{et al.}}:
Lineage-Aware Memory Governance for Enterprise AI Agents}

\begin{abstract}
Enterprise AI agents that share a memory store face two unaddressed
risks: sensitive data can leak through legitimately computed results the
requester could not derive, and departments can silently compute a
same-named key performance indicator (KPI) through conflicting logic.
Existing agent-memory systems (e.g., MemGPT, Zep, A-MEM) gate retrieval
by content, ownership, and role, not derivation, missing a cached insight
that embeds a forbidden column. We introduce the Analytical Memory Unit
(AMU), a memory schema that attaches a full derivation (lineage) graph to
every cached result, gated by a retrieval policy that serves a hit only
when the requester is authorised for every column touched. Provided
lineage recording is complete, we prove by construction that the policy
blocks retrieval of results derived from a sensitive column outside the
requester's permissions, at $O(n)$ worst case -- a conditional design
guarantee, not an empirical claim, that excludes derived features
encoding sensitive information without naming their source. Eliminating
measured leakage required $75$--$90\%$ recorded lineage completeness, so
we treat $90\%$ as a conservative deployment target. Across six
experiments, lineage-gated retrieval removes the $18.8$--$25.5\%$
cross-department leakage naive content-gated memory suffers, keeping
$81.5$--$82.6\%$ of memory reuse at $13.8\,\mu$s worst-case overhead. A
real-agent proof-of-concept with LLM-generated SQL is consistent with the
guarantee: zero leaks over 9 round-trips, two conflicts caught
automatically -- though a feasibility demonstration, not evidence of
production viability. This offers a practical governance layer for
shared agent memory, complementing source-layer access control and
supporting EU AI Act compliance.
\end{abstract}


\maketitle

\section{Introduction}
\label{sec:intro}

Consider a mid-sized enterprise running separate Finance and Marketing AI
agents against a shared customer data warehouse.
Finance joins \texttt{customer\_pii} (which contains \texttt{income}) against
\texttt{transactions} to identify high-value customers, producing a segment
of 12,400 customer IDs and a count.
The result itself carries no raw personally identifiable information (PII)
and no sensitivity tag. So under any memory system that governs access by
\emph{content}, it gets freely cached and served to the next agent that
asks for ``high-value customers'' -- including Marketing, whose role was
never granted access to \texttt{income}. The derivation path quietly
violated column-level access policy, and none of the content-gated
safeguards common in current memory-sharing systems (\S\ref{sec:related})
would catch it, because the violation is invisible in the stored value
and shows up only in \emph{how the value was computed}. This is simply
how the content-gated memory-sharing pattern behaves in every system we
survey in \S\ref{sec:related}.

\subsection{Motivation and Prior Gap}

The urgency tracks the pace of agent adoption: Gartner data shows 80\% of
enterprise applications shipped or updated in Q1 2026 embed at least one
AI agent, up from 33\% in 2024, and 31\% of enterprises already run at
least one agent in production, rising to 47\% in banking and
insurance~\cite{digitalapplied2026adoption}.
IBM's 2025 Cost of a Data Breach Report puts the global average breach cost
at \$4.44 million (\$10.22 million in the US)~\cite{ibm2025breach}, and El
Yagoubi et al.'s AgentLeak benchmark finds that inter-agent messages leak
sensitive data at $68.8\%$, compared with $27.2\%$ for final outputs alone,
meaning output-only auditors---who inspect only final
outputs---miss $41.7\%$ of privacy
violations~\cite{elyagoubi2026agentleak}.
Systems such as MemGPT~\cite{packer2023memgpt}, Zep~\cite{rasmussen2025zep},
A-MEM~\cite{xu2025amem}, Governed Memory~\cite{taheri2026governed}, SSGM
\cite{lam2026ssgm}, and Oracle AI Agent Memory~\cite{oracle2026memory}
show that shared, governed memory cuts down redundant computation. The
problem is that \textbf{every current system controls access at the level
of \emph{content} and \emph{tags}, not \emph{derivation}}
(\S\ref{sec:related}, Table~\ref{tab:comparison}): a memory entry gets labelled with an
owner, sensitivity class, or role, and retrieval is gated on whether the
requester holds the matching label. None of these systems, and none of
the source-layer controls such as Apache Ranger or Databricks Unity
Catalog, record which tables were joined, which columns were touched, or
what filter logic produced a stored result. Source-layer controls stop an
unauthorised agent from \emph{querying} a restricted column, but they do
nothing once a legitimately computed result is cached and later served to
a different, less-privileged requester.
There is a second, independent failure mode on top of this: departments
often compute the same-named key performance indicator (KPI) through
different business logic (e.g., a 90-day vs.\ 60-day churn window), and
without derivation tracking, shared memory quietly propagates one team's
definition to the other. This is documented in the business intelligence
(BI) literature as semantic drift~\cite{colrows2024metrics}, and it is apparently serious
enough that Snowflake, Salesforce, dbt Labs, BlackRock, and RelationalAI
jointly launched the Open Semantic Interchange initiative in
2025~\cite{snowflake2025osi} -- though that effort targets human-facing BI
tools, not autonomous agents writing their own queries.
There is also a regulatory angle worth noting: the EU AI Act imposes
explainability and risk-management obligations on high-risk AI systems for
which data lineage is a structural prerequisite~\cite{euaiact2024}. A
memory layer that can actually \emph{prove} rather than just assert that
no result crossed a column-permission boundary turns what would otherwise
be an audit question into something the system itself can check
mechanically.

\subsection{Contributions}

We make four contributions in this paper. First is the \emph{Analytical
Memory Unit} (AMU) schema and the \emph{lineage-gated retrieval} policy
that goes with it: a memory schema that attaches a full derivation graph
to every stored result, and a retrieval algorithm that gates memory hits
on column-level derivation permissions rather than content tags
(\S\ref{sec:design}). Second, we give a formal safety guarantee --
Theorem~\ref{thm:safety} proves that lineage-gated retrieval prevents
leakage of any sensitive column outside the requester's permissions under
complete lineage recording, with corollaries covering unsafe store
contents and policy updates, plus $O(n)$/$O(1)$ complexity analysis
(\S\ref{sec:formal}). Third is an eight-threat model (T1--T8) that lays out
what the mechanism protects against, what it does not, and what production
control we recommend for each gap (\S\ref{sec:threat}). Fourth, we run six
statistically validated experiments: two schema configurations, a
completeness-degradation study, an extended 43-pair fuzzy
conflict-detection study, microsecond runtime benchmarks, and a real-agent
integration with automatic SQL lineage extraction. Together these quantify
the governance-efficiency tradeoff and ground the mechanism in a real,
open-source agent stack (\S\ref{sec:experiments}).
Two boundaries on this protection are worth stating here rather than only
in \S\ref{sec:discussion}. The safety guarantee holds only to the extent
lineage recording is complete (Assumption~\ref{ass:complete}); Experiment~3
shows measured leakage eliminated only once recorded completeness reaches
$75$--$90\%$ depending on schema, and our own automatic extractor has
known blind spots for views, stored procedures, and deep aliasing. And it
does not cover threat T4 (\S\ref{sec:threat}), inference through undeclared
derived features -- in our view the most consequential gap for production
use.
We keep the main text within IEEE Access length guidelines by moving full
derivations, line-by-line algorithm walkthroughs, extended related-work
discussion, and complete per-experiment tables to the Supplementary
Material.

\section{Related Work}
\label{sec:related}

\textbf{Agent memory systems.}
Du~\cite{du2026survey} formalises agent memory as a write-manage-read loop.
MemGPT~\cite{packer2023memgpt} pages an LLM between a bounded in-context
tier and an external archival/recall store, moving data across tiers via
explicit LLM-issued function calls -- a capacity-management mechanism
with no notion of what produced a stored value.
Zep~\cite{rasmussen2025zep} instead extracts entities and relations from
agent interactions into a bi-temporal knowledge graph (separate
valid-time and transaction-time edges) and retrieves via graph traversal
and semantic search, so its notion of provenance is \emph{when a fact was
true and when it was recorded}, not which source columns produced it.
A-MEM~\cite{xu2025amem} generates keywords, tags, and contextual
descriptions for each note with an LLM and links related notes
Zettelkasten-style, so its retrieval mechanism operates on note-to-note
semantic similarity rather than on a computation's inputs.
All three are state of the art for conversational or task-level memory,
but their gating and retrieval mechanisms all operate above the level of
individual source columns, so none model the provenance of a computed
analytical result.

\textbf{Governed and enterprise memory.}
Governed Memory~\cite{taheri2026governed} and SSGM~\cite{lam2026ssgm}
introduce role-based access, tiered governance routing, and consistency
verification for evolving memory stores. Oracle AI Agent Memory
\cite{oracle2026memory} provides a commercially deployed governed core, and
AgentGuardian~\cite{abaev2026agentguardian} learns access-control policies
from execution traces. What all four have in common is that they govern
access to a memory entry \emph{as a whole} -- by owner, role, or tool-call
legitimacy -- rather than by the individual columns that contributed to
it. That means none of them can tell that a permitted entry embeds a
forbidden column.

\textbf{Multi-user and multi-principal shared memory.}
Two very recent systems tackle adjacent problems. Collaborative
Memory~\cite{rezazadeh2025collabmem} models multi-user, multi-agent memory
sharing via bipartite user--agent--resource access graphs with per-fragment
provenance (contributing agents, accessed resources, timestamps) for
retrospective permission checks, and it proves adherence to time-varying
read/write policies. Like Governed Memory and SSGM, though, its policies
gate access to a memory \emph{fragment} as an atomic unit by identity and
permission graph rather than by the individual source columns a computed
result embeds, so a fragment it permits can still quietly carry a value
derived from a column outside the requester's grant.
GateMem~\cite{ren2026gatemem} is a benchmark rather than a memory system:
it evaluates existing memory-agent baselines for utility, access control,
and active forgetting across multi-principal shared-memory settings
(medical, office, education, and household domains), and finds that no
evaluated method achieves robust access control without sacrificing
utility. That is a useful independent confirmation of the governance gap
this paper targets, though its leak-target annotations are defined at the
level of \emph{shared facts} rather than \emph{derivation graphs}, so it
does not test column-level lineage gating specifically. Neither system
attaches a derivation graph to a cached analytical result or gates
retrieval on the columns touched to produce it, which is the specific
mechanism AMU contributes.

\textbf{Column-level security and provenance.}
Apache Ranger, Databricks Unity Catalog, and BigQuery enforce column-gated
access at the \emph{query layer}; DePLOI~\cite{subramaniam2024deploi} audits
such policies via natural-language-to-SQL (NL2SQL) translation.
The W3C PROV standard~\cite{moreau2011prov} and Herschel et
al.~\cite{herschel2017survey} model data provenance for extract-transform-load
(ETL) pipelines, and
LINEAGEX~\cite{zhang2025lineagex} extracts column-level SQL lineage with
high accuracy.
None of these fire at the point where a legitimately computed result gets
cached and later served to a different, less-privileged requester. That is
the gap our AMU schema closes, by applying column-level derivation
tracking to \emph{agent memory entries} rather than raw database records.

\textbf{Privacy and security in multi-agent systems.}
AgentLeak~\cite{elyagoubi2026agentleak} shows most multi-agent leakage
occurs through inter-agent channels; OMNI-LEAK~\cite{naik2026omnileak} and
Alizadeh et al.~\cite{alizadeh2025leak} show prompt injection can exfiltrate
observed data; AuthGraph~\cite{wang2026authgraph} defends individual agents
by comparing an execution-derived provenance graph against an authorisation
graph (closest in spirit to our column-level gate, but scoped to
single-agent tool-calling rather than cross-department memory retrieval).
This literature targets communication channels and single-agent
manipulation. Our work targets a complementary vector: a \emph{shared
analytical memory store} that leaks even when every communication link and
agent is individually secure (threat T8, \S\ref{sec:threat}).

\textbf{Metric-definition governance.}
Semantic-layer tools and the Open Semantic Interchange
initiative~\cite{snowflake2025osi} centralise \emph{human}-defined metrics
but provide no mechanism for detecting drift when an autonomous agent
computes a KPI through its own generated query---the gap our
\texttt{definition\_hash} conflict detector (Algorithm~\ref{alg:write})
targets.

Table~\ref{tab:comparison} summarises six governance dimensions across
representative systems: no prior system tracks column-level derivation, and
consequently none can gate retrieval on whether the \emph{requester} is
permitted to see every column that contributed to a cached result.
Extended per-system discussion appears in the Supplementary Material.

\begin{table*}[t]
\centering
\caption{Feature comparison of representative agent-memory and governance
systems against the Analytical Memory Unit (AMU). ``Role ACL'' = role-based
access-control list. \ding{51}~= supported,
\ding{55}~= not supported, \textasciitilde~= partial/indirect support.
Cells reflect each system's \emph{documented design decisions}, not
independently verified performance; see the note on source provenance
following this table.}
\label{tab:comparison}
\footnotesize
\begin{tabular}{@{}p{0.26\linewidth}p{0.11\linewidth}p{0.1\linewidth}p{0.11\linewidth}p{0.1\linewidth}p{0.11\linewidth}@{}}
\toprule
\textbf{System} & \textbf{Content tag} & \textbf{Role ACL} & \textbf{Column lineage} & \textbf{Drift detect} & \textbf{Retrieval gate} \\
\midrule
MemGPT             & \ding{55} & \ding{55} & \ding{55} & \ding{55} & \ding{55} \\
Zep                 & \ding{51} & \ding{55} & \ding{55} & \ding{55} & \ding{55} \\
A-MEM               & \ding{51} & \ding{55} & \ding{55} & \ding{55} & \ding{55} \\
Governed Memory     & \ding{51} & \ding{51} & \ding{55} & \ding{55} & \ding{51} \\
SSGM                & \ding{51} & \textasciitilde & \ding{55} & \textasciitilde & \ding{51} \\
Collaborative Memory & \textasciitilde & \ding{51} & \ding{55} & \ding{55} & \ding{51} \\
Oracle AI Agent Memory & \ding{51} & \ding{51} & \ding{55} & \ding{55} & \ding{51} \\
AgentGuardian       & \ding{55} & \ding{51} & \ding{55} & \ding{55} & \textasciitilde \\
Source-layer ACL (Ranger/Unity Catalog) & \ding{55} & \ding{51} & \textasciitilde & \ding{55} & \ding{55} \\
\textbf{AMU (this work)} & \ding{51} & \ding{51} & \textbf{\ding{51}} & \textbf{\ding{51}} & \ding{51} \\
\bottomrule
\end{tabular}
\end{table*}

\textbf{A note on source provenance.} The comparison systems in
Table~\ref{tab:comparison} sit at very different stages of peer review, and
that difference matters for how the comparison should be read. A-MEM~\cite{xu2025amem} appeared at NeurIPS 2025.
MemGPT~\cite{packer2023memgpt} is unpublished but widely adopted and
highly cited, and it established the field's dominant memory-tiering
paradigm anyway. Zep~\cite{rasmussen2025zep}, Governed
Memory~\cite{taheri2026governed}, SSGM~\cite{lam2026ssgm}, and
AgentGuardian~\cite{abaev2026agentguardian} are, at the time of writing,
preprints describing production or proposed architectures rather than
independently peer-reviewed evaluations (AgentGuardian, from a
Ben-Gurion University security group, does report its own empirical
evaluation on two real-world agent applications, but that evaluation
hasn't been peer-reviewed either). Oracle AI Agent
Memory~\cite{oracle2026memory} is documented only in a vendor developer
blog, with no publicly available technical paper. Collaborative
Memory~\cite{rezazadeh2025collabmem}, from Accenture's Center for Advanced
AI, and GateMem~\cite{ren2026gatemem} are, like most of this comparison,
unreviewed preprints as of this writing. We include all of them because
Table~\ref{tab:comparison} compares \emph{documented design
decisions} -- what each system's own description says it tracks and gates
on -- rather than reported performance numbers we did not independently
verify. The governance gap we identify (content/tag-level access control
rather than derivation-level) is a structural property of each system's
publicly stated design, and that holds regardless of whether the design
has cleared peer review. We are not making comparative performance claims
from these sources, only feature-support claims, and we want to flag that
distinction explicitly so readers can weigh the comparison accordingly.

\section{System Design}
\label{sec:design}

\subsection{The Analytical Memory Unit}

\begin{definition}[Analytical Memory Unit]
An \emph{Analytical Memory Unit} (AMU) is a tuple
$\langle \mathit{metric}, v, d, \tau, L \rangle$ where:
$\mathit{metric}$ is a metric-name string;
$v \in \mathbb{R}$ is the computed result value;
$d$ is the producing department;
$\tau$ is an integer epoch (supports time-to-live (TTL) eviction); and
$L = (s_1, \ldots, s_k)$ is a \emph{lineage graph}, each step
$s_i = (\mathit{table}_i,\, C_i,\, \varphi)$
where $C_i$ is the set of columns of $\mathit{table}_i$ accessed and
$\varphi$ is a filter-logic string.
\end{definition}

\begin{figure}[t]
\centering
\begin{tikzpicture}[
  font=\fontsize{5.6}{6.4}\selectfont,
  leaf/.style={draw=algframe,thick,fill=algbg,rounded corners=1.5pt,
               align=center,text width=2.05cm,minimum height=0.62cm,
               inner sep=2pt},
  joinbox/.style={draw=algframe,thick,fill=white,rounded corners=1.5pt,
               align=center,text width=3.55cm,inner sep=3pt},
  resbox/.style={draw=algframe,line width=0.9pt,fill=blue!6,
               rounded corners=1.5pt,align=center,text width=2.3cm,
               inner sep=3pt},
  tagbox/.style={draw=algframe,thick,dashed,fill=algbg,
               rounded corners=1.5pt,align=left,text width=2.05cm,
               inner sep=3pt},
  arr/.style={-{Latex[length=4.5pt,width=3.5pt]},thick,draw=algframe}
]

\node[leaf] (cust) at (0,0)
  {\texttt{customers}\\\{customer\_id,\\signup\_date\}};
\node[leaf] (tx) at (2.35,0) {\texttt{transactions}};
\node[leaf] (pii) at (4.7,0)
  {\texttt{customer\_pii}\\\{last\_txn\_date,\\\textbf{income}\}};

\node[joinbox] (join) at (3.525,1.35)
  {\texttt{transactions} $\bowtie$ \texttt{customer\_pii}\\[2pt]
   \textit{filter:} last\_txn\_date $<$ now()$-$90d\\ AND income $>$ 20000};

\node[resbox] (churn) at (1.05,2.85)
  {\texttt{churn\_rate} $=$ \textbf{0.114}\\(Finance, epoch~17)};

\node[tagbox] (tags) at (4.15,2.85)
  {$\mathrm{S}(a) = \{$\textbf{income}$\}$\\[2pt]
   $\mathrm{H}(a) =$ \texttt{7a41c9e0}\\\texttt{d3f2}};

\draw[arr] (tx.north) |- ($(tx.north)+(0,0.25)$) -| (join.south -| tx);
\draw[arr] (pii.north) |- ($(pii.north)+(0,0.25)$) -| (join.south -| pii);
\draw[arr] (cust.north) to[out=90,in=200] (churn.west);
\draw[arr] (join.north) to[out=90,in=-20] (churn.east);
\draw[arr,dashed] (churn.east) to (tags.west);
\end{tikzpicture}
\caption{Worked lineage graph: Finance's \texttt{churn\_rate} is derived by
joining three source tables. Because the join touches the sensitive column
\texttt{income}, the derived tag $\mathrm{S}(a)$ makes this single cached
number unsafe to serve to any department lacking \texttt{income}
access---even though the value $0.114$ alone reveals nothing.}
\label{fig:lineage-example}
\end{figure}
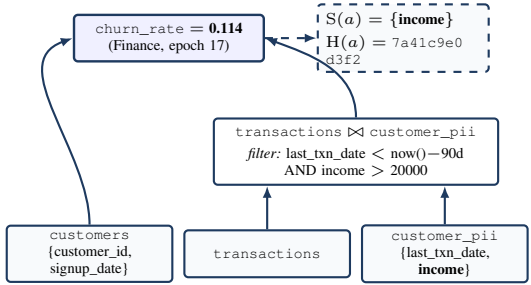

Figure~\ref{fig:lineage-example} shows a worked example: the result value
$v$ alone carries no information about how it was produced; the lineage
graph $L$ makes the derivation explicit and machine-checkable at retrieval
time.
Two quantities are derived once at AMU creation and cached:
\begin{align}
  \mathrm{S}(a) &= \bigcup_{i=1}^{k} C_i \;\cap\; \mathcal{S}
  \label{eq:sensitivity} \\[2pt]
  \mathrm{H}(a) &= \mathrm{SHA256}\!\left(
    \mathrm{sort}(\mathit{tables}) \,\|\,
    \mathrm{sort}\!\left(\bigcup_i C_i\right) \,\|\, \varphi
  \right)_{\![0{:}12]}
  \label{eq:defhash}
\end{align}
where $\mathcal{S}$ is the sensitive-column registry and $\mathrm{H}(a)$ is
the \emph{definition hash} fingerprinting \emph{how} the metric was
computed.
Both are sub-microsecond in practice (\S\ref{sec:runtime}).
Table~\ref{tab:amu} lists the full schema.

\begin{table}[t]
\centering
\caption{Analytical Memory Unit (AMU) schema fields.}
\label{tab:amu}
\footnotesize
\begin{tabular}{@{}p{0.29\linewidth}p{0.63\linewidth}@{}}
\toprule
\textbf{Field} & \textbf{Description} \\
\midrule
\texttt{metric\_name}      & Business KPI name (\emph{e.g.}, \texttt{churn\_rate}) \\
\texttt{value}, \texttt{owner\_dept}, \texttt{epoch} & Result, producing department, TTL time index \\
\texttt{lineage}           & Derivation graph: sequence of $(t, C, \varphi)$ steps \\
\texttt{sensitivity\_tags} & $\mathrm{S}(a)$ (Eq.~\ref{eq:sensitivity}), cached \\
\texttt{definition\_hash}  & $\mathrm{H}(a)$ (Eq.~\ref{eq:defhash}), cached \\
\bottomrule
\end{tabular}
\end{table}

\FloatBarrier
\subsection{Lineage-Gated Retrieval and Conflict Detection}
\label{sec:algorithms}

Algorithms~\ref{alg:retrieve} and~\ref{alg:write} give the two core
operations; $\mathrm{P}(d)$ denotes department $d$'s permitted column set.
A full walkthrough of Algorithm~\ref{alg:retrieve} follows below; the
analogous walkthrough of Algorithm~\ref{alg:write} appears in the
Supplementary Material.

\begin{algorithmbox}[label={alg:retrieve}]{Lineage-Gated Retrieve}{%
  \textnormal{Input: \textit{store, metric, dept, fresh}; Output: \textit{result}; \textit{leaked}=\textbf{False}}}
candidates <- store[metric]         % all stored AMUs for this metric
permitted  <- P(dept)               % column set the department may see
blocked    <- False

for each a in REVERSE(candidates):  % most-recent-first
  if S(a) not_subset permitted:     % derivation touches a forbidden column
    blocked <- True
    continue                        % skip this candidate; try the next
  end if
  return (reused=True, value=a.value, leaked=False, blocked=False)
end for

store[metric].append(fresh)
return (reused=False, value=fresh.value, leaked=False, blocked=blocked)
\end{algorithmbox}

\begin{algorithmbox}[label={alg:write}]{Lineage-Aware Write}{%
  \textnormal{Input: \textit{store}, AMU \textit{a}; Output: \textit{conflict} \textbf{bool}}}
existing <- store[a.metric]   % all AMUs stored for this metric name
conflict <- False

for each e in existing:
  if e.dept != a.dept:        % cross-department pair
    if H(e) != H(a):          % definition hashes differ -> potential conflict
      conflict <- True
      break
    end if
  end if
end for

store[a.metric].append(a)     % always persist; conflict is surfaced, not blocked
return conflict
\end{algorithmbox}

\textbf{Line-by-line walkthrough of Algorithm~\ref{alg:retrieve}.}
Line~1 retrieves every AMU ever stored under this metric name; this list is
the candidate pool. Line~2 resolves the requester's column permission set
once, outside the loop, so it is not recomputed per candidate. Line~5
iterates candidates \emph{most-recent-first}: this is a deliberate policy
choice, not an implementation detail---it means that when a safe candidate
exists, the department receives the freshest safe value rather than an
arbitrary historical one. Line~6 is the gate: a single $O(|\mathcal{S}|)$
set-difference check (cached $\mathrm{S}(a)$ against $\mathrm{P}(d)$). If the
check fails, line~7 records that at least one candidate was blocked (used
only for diagnostics and the experiments in \S\ref{sec:experiments}) and
line~8 continues to the next-oldest candidate---the rejected candidate is
never returned, never partially returned, and never logged with its value.
Line~10 is the only successful-return statement in the loop body: it fires
the first time a candidate's derivation is fully contained in the
requester's permissions. If the loop exhausts all candidates without a safe
hit, line~13 appends the caller-supplied, already-in-scope \textit{fresh}
computation to the store (so future requests may reuse it) and line~14
returns it. Every exit path is covered by exactly one of these two return
statements, which is what the proof of Theorem~\ref{thm:safety} exploits.
\emph{Complexity:} the loop runs at most $n$ times ($n$ = candidates for
this metric), each iteration doing $O(|\mathcal{S}|)$ work, giving
$O(n\cdot|\mathcal{S}|)$ worst case; since $|\mathcal{S}|$ is a small
constant per organisation, we report this as $O(n)$ in
Table~\ref{tab:complexity}. The best case (first candidate safe) is
$O(|\mathcal{S}|) = O(1)$.

Algorithm~\ref{alg:write} compares cached definition hashes $\mathrm{H}(\cdot)$
rather than raw lineage graphs, making conflict detection $O(1)$ per pair; the
write path never blocks on a conflict, only surfaces one for downstream
alerting (threat T2). Its full line-by-line walkthrough, structurally
analogous to the one above, is given in the Supplementary Material.

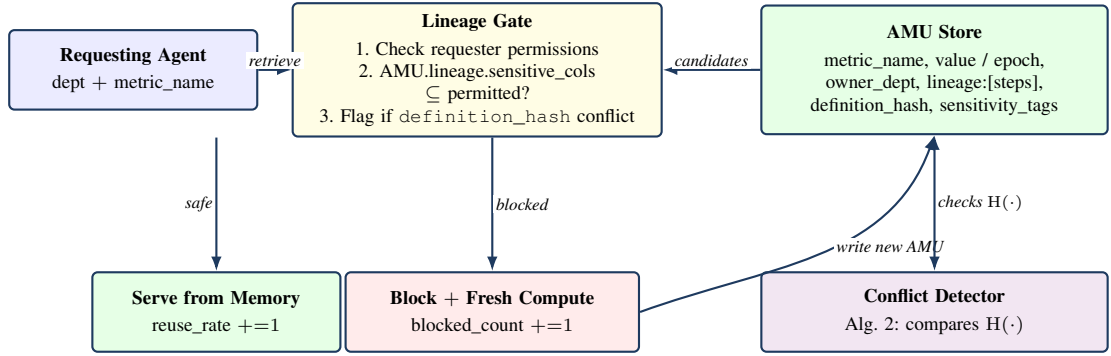
\begin{figure*}[t]
\centering
\begin{tikzpicture}[
  font=\fontsize{7.2}{8.4}\selectfont,
  every node/.style={align=center},
  agent/.style={draw=algframe,thick,fill=blue!8,rounded corners=2pt,
                text width=3.0cm,minimum height=1.05cm,inner sep=4pt},
  gate/.style={draw=algframe,thick,fill=yellow!12,rounded corners=2pt,
                text width=4.6cm,minimum height=1.7cm,inner sep=4pt},
  store/.style={draw=algframe,thick,fill=green!10,rounded corners=2pt,
                text width=4.3cm,minimum height=1.7cm,inner sep=4pt},
  serve/.style={draw=algframe,thick,fill=green!10,rounded corners=2pt,
                text width=3.0cm,minimum height=1.05cm,inner sep=4pt},
  block/.style={draw=algframe,thick,fill=red!8,rounded corners=2pt,
                text width=3.6cm,minimum height=1.05cm,inner sep=4pt},
  conflict/.style={draw=algframe,thick,fill=violet!10,rounded corners=2pt,
                text width=4.3cm,minimum height=1.05cm,inner sep=4pt},
  lbl/.style={font=\fontsize{6.4}{7.4}\selectfont\itshape,fill=white,
              inner sep=1pt},
  arr/.style={-{Latex[length=5.5pt,width=4pt]},thick,draw=algframe}
]

\node[agent]   (agent)  at (0,3.55)   {\textbf{Requesting Agent}\\[2pt] dept $+$ metric\_name};
\node[gate]    (gate)   at (4.55,3.55) {\textbf{Lineage Gate}\\[2pt]
  1.~Check requester permissions\\
  2.~AMU.lineage.sensitive\_cols $\subseteq$ permitted?\\
  3.~Flag if \texttt{definition\_hash} conflict};
\node[store]   (store)  at (10.6,3.55) {\textbf{AMU Store}\\[2pt]
  metric\_name, value / epoch,\\ owner\_dept, lineage:[steps],\\
  definition\_hash, sensitivity\_tags};

\node[serve]   (serve)  at (1.1,0.35) {\textbf{Serve from Memory}\\[2pt] reuse\_rate $+{=}1$};
\node[block]   (block)  at (4.75,0.35) {\textbf{Block $+$ Fresh Compute}\\[2pt] blocked\_count $+{=}1$};
\node[conflict](conf)   at (10.6,0.35) {\textbf{Conflict Detector}\\[2pt] Alg.~\ref{alg:write}: compares $\mathrm{H}(\cdot)$};

\draw[arr] (agent.east) -- node[lbl,above]{retrieve} (gate.west);
\draw[arr] (store.west) -- node[lbl,above]{candidates} (gate.east);
\draw[arr] (gate.south -| serve) -- node[lbl,left]{safe} (serve.north);
\draw[arr] (gate.south -| block) -- node[lbl,right]{blocked} (block.north);
\draw[arr] (block.east) to[out=20,in=250] node[lbl,right,pos=0.55]{write new AMU} (store.south);
\draw[arr,{Latex[length=5.5pt,width=4pt]}-{Latex[length=5.5pt,width=4pt]}]
  (store.south) -- node[lbl,right]{checks $\mathrm{H}(\cdot)$} (conf.north);
\end{tikzpicture}
\caption{AMU architecture. An agent's executed query is intercepted at the
database cursor and its lineage extracted (left); the AMU is written to the
shared store, where Algorithm~\ref{alg:write} checks $\mathrm{H}(\cdot)$
against cross-department entries. On the read path (right), the lineage gate
checks $\mathrm{S}(a) \subseteq \mathrm{P}(d)$ before serving from memory;
unsafe candidates are blocked and the agent falls back to a fresh in-scope
computation, itself written back through the same path.}
\label{fig:arch}
\end{figure*}

\subsection{Implementation Considerations}
\label{sec:implementation}

Three practical points come up once you actually try to deploy this
(details in the Supplementary Material). The cleanest integration point is
the database driver layer: if you wrap the cursor so every executed SQL
string passes through a lineage extractor (\textsc{SQLGlot}, Experiment~6),
lineage capture needs zero cooperation from the agent. The system also
has to \emph{fail closed} -- if lineage extraction fails, treat the AMU as
maximally sensitive ($\mathrm{S}(a) \leftarrow \mathcal{S}$) rather than
failing open, and if a department's permission set is unavailable, default
retrieval to denial rather than access. Finally, the \texttt{epoch} field
does double duty: it supports ordinary TTL eviction, but it can also drive
\emph{policy-driven} invalidation -- when $\mathcal{S}$ or $\mathrm{P}(d)$
changes, cached $\mathrm{S}(a)$ values should be lazily recomputed on the
next retrieval (Corollary~\ref{cor:policy-update}, threat T3).

\section{Formal Analysis}
\label{sec:formal}

\begin{assumption}[Complete Lineage]
\label{ass:complete}
For every AMU $a$ in the store, $a.\mathit{lineage}$ records every table,
column, and filter predicate of the computation that produced $a.\mathit{value}$.
\end{assumption}

\begin{theorem}[Sensitive-Column Safety]
\label{thm:safety}
Under Assumption~\ref{ass:complete}, Algorithm~\ref{alg:retrieve} returns no
result to department $d$ that was derived from a sensitive column
$c \in \mathcal{S}$ with $c \notin \mathrm{P}(d)$.
\end{theorem}

\begin{remark}
This is a conditional design-correctness result, not an empirical one: it
follows from the structure of Algorithm~\ref{alg:retrieve} together with
Assumption~\ref{ass:complete}, by the case analysis below. Wherever a
$0\%$ leak rate appears later in this paper (\S\ref{sec:experiments}), it
is a consequence of this proof holding in the setting tested, not an
independent empirical discovery -- the empirical content of those
experiments is the naive-baseline leak rate, the reuse-rate cost, and how
performance degrades once Assumption~\ref{ass:complete} is violated
(\S\ref{sec:degradation}).
\end{remark}

\begin{proof}
By inspection, Algorithm~\ref{alg:retrieve} has exactly two statements that
return a value to the caller: line~10 (memory-served) and line~14
(fallback). Every invocation terminates by reaching one of the two, since
the \texttt{for} loop either returns early at line~10 or exhausts
\textit{candidates} and falls through to lines~13--14. We show the returned
AMU's derivation touches no forbidden sensitive column on each path.

\smallskip\noindent\textbf{Case 1: memory-served path (line~10).}
Reaching line~10 for candidate $a$ requires the guard at line~6 to have
evaluated false, i.e.\ $\mathrm{S}(a) \subseteq \mathrm{P}(d)$ holds for that
$a$. Unfolding Equation~\ref{eq:sensitivity}, $\mathrm{S}(a) =
\bigl(\bigcup_{i=1}^{k} C_i\bigr) \cap \mathcal{S}$, so $\mathrm{S}(a)
\subseteq \mathrm{P}(d)$ is equivalent to: for every step $s_i \in
a.\mathit{lineage}$ and every column $c \in C_i \cap \mathcal{S}$, $c \in
\mathrm{P}(d)$. Contrapositively, there is no sensitive column $c \in
\mathcal{S}$ with $c \notin \mathrm{P}(d)$ appearing in any $C_i$---i.e., no
such $c$ appears anywhere in $a$'s derivation. Since
Assumption~\ref{ass:complete} guarantees $a.\mathit{lineage}$ records
\emph{every} column touched by the computation that produced $a.value$,
this is precisely the safety property claimed for the value returned at
line~10.

\smallskip\noindent\textbf{Case 2: fallback path (line~14).}
The input precondition to Algorithm~\ref{alg:retrieve} (stated in the
signature) requires that \textit{fresh} is computed entirely within
$\mathrm{P}(d)$ by the caller, before Algorithm~\ref{alg:retrieve} is
invoked. Under Assumption~\ref{ass:complete}, \textit{fresh}.\textit{lineage}
records every column of that computation, so every column recorded is, by
the precondition, in $\mathrm{P}(d)$; in particular no sensitive column
outside $\mathrm{P}(d)$ is touched.

\smallskip\noindent
Both cases establish that the AMU returned to department $d$ has a
derivation containing no sensitive column outside $\mathrm{P}(d)$, for an
arbitrary invocation with arbitrary \textit{store} contents. Since the two
cases are exhaustive over all return statements, the claim holds for every
invocation of Algorithm~\ref{alg:retrieve}.
\end{proof}

\begin{corollary}[Robustness to unsafe store contents]
\label{cor:unsafe-contents}
The guarantee holds even when \textit{store} contains AMUs whose derivation
touches columns outside $\mathrm{P}(d)$: such candidates fail the gate and
are skipped, regardless of how many exist or where they occur.
\end{corollary}

\begin{corollary}[Behaviour under policy updates]
\label{cor:policy-update}
The guarantee holds \emph{per policy version}: if $\mathrm{S}(a)$ is cached
under registry $\mathcal{S}_1$ and the organisation transitions to
$\mathcal{S}_2$ without recomputing existing AMUs, a column newly added to
$\mathcal{S}_2$ may not appear in the stale cached tag. The guarantee is
restored for all subsequent retrievals once $\mathrm{S}(a)$ is recomputed
against $\mathcal{S}_2$ for every AMU predating the change (formalising
threat T3).
\end{corollary}

\subsection{Complexity and Storage}
\label{sec:complexity}

Let $n$ = AMUs stored per metric, $k$ = lineage steps, $c$ = columns per
step.
Table~\ref{tab:complexity} reports asymptotic and measured costs.

\begin{table}[t]
\centering
\caption{Asymptotic complexity and measured latency (means over 2{,}000
iterations; $k{=}2$--$3$, $c{=}4$--$7$). Full runtime sweep in the
Supplementary Material.}
\label{tab:complexity}
\footnotesize
\begin{tabular}{@{}p{0.32\linewidth}p{0.12\linewidth}p{0.13\linewidth}p{0.13\linewidth}@{}}
\toprule
\textbf{Operation} & \textbf{Cplx.} & \textbf{$n{=}1$} & \textbf{$n{=}50$} \\
\midrule
Gate-check predicate        & $O(|\mathcal{S}|)$ & $0.09\,\mu$s & $0.09\,\mu$s \\
Write --- conflict detection & $O(n)$  & $0.28\,\mu$s & $1.63\,\mu$s \\
Retrieve --- worst case      & $O(n)$  & $0.66\,\mu$s & $13.82\,\mu$s \\
Retrieve --- expected case   & $O(1)$  & $0.66\,\mu$s & $0.64\,\mu$s \\
\bottomrule
\end{tabular}
\end{table}

The expected-case retrieve latency is constant regardless of store size;
the worst case (all $n$ candidates blocked) scales linearly, reaching
$13.82\,\mu$s at $n{=}50$, but in a realistic production store $n \leq 5$
(at most one AMU per department variant per metric), so both paths are
effectively $O(1)$ and well below $2\,\mu$s---negligible relative to any
real analytics query (ms--s).
We chose the worst-case $O(n)$ bound deliberately -- it is not something we
overlooked. Algorithm~\ref{alg:retrieve} scans most-recent-first so it
returns the freshest \emph{safe} value rather than an arbitrary one. An
alternative design that indexed AMUs by $(\mathit{metric}, \mathrm{S}(a))$
could reach $O(1)$ worst case, but at the cost of extra index-maintenance
complexity and giving up the recency-preference guarantee. We passed on
that design because $n$ is small by construction: a metric name
accumulates at most one AMU per distinct (department, lineage-variant)
combination, and that held for $n \leq 5$ in every workload we observed
(Experiments~1--2 and~6), where the $O(n)$ worst case costs at most
$1.74\,\mu$s -- immaterial in practice. We still report the asymptotic
bound honestly, though, because it would matter under a poorly bounded TTL
policy, which is exactly why we recommend TTL eviction as a standing
operational control.
Relative to the $O(1)$ dictionary lookup used by all prior
systems~\cite{packer2023memgpt,rasmussen2025zep,xu2025amem,taheri2026governed,lam2026ssgm,oracle2026memory},
lineage-gated retrieval trades at most $13.82\,\mu$s in absolute terms for
the safety guarantee of Theorem~\ref{thm:safety}---a tradeoff Experiments~1--2
show buys elimination of an $18.8$--$25.5\%$ structural leak rate at a
$13$--$14$ percentage-point reuse cost.
Storage overhead is modest: an AMU adds $k$ tuples of table/column names,
the filter string, and a 12-character hash---roughly $4$--$8\times$ a bare
value entry (200--600 bytes for a typical 2--4 table join), smaller than the
SQL query most systems already persist alongside the result.

\section{Threat Model}
\label{sec:threat}

We model principals as enterprise departments operating AI agents that are
\emph{honest-but-potentially-over-privileged}: they request metrics they
have a legitimate business need for, but not all derivation paths fall
within their column permissions.
Two threats motivated this work in the first place, and they are the ones
the mechanism directly addresses: T1, where an agent retrieves a cached
result derived from a sensitive column it should not see (mitigated by the
$\mathrm{S}(a)\subseteq\mathrm{P}(d)$ gate), and T2, where two departments
store conflicting definitions of the same-named KPI (surfaced by the
definition-hash conflict detector).
The most important gap we do \emph{not} address is T4, inference attacks:
a derived feature that \emph{encodes} sensitive information without naming
the source column (say, a risk score computed from \texttt{income}) will
slip past the gate unless the dependency is explicitly recorded in the
lineage graph.
Table~\ref{tab:threats} summarises all eight threats; full prose
descriptions of each are in the Supplementary Material.

\begin{table*}[t]
\centering
\caption{Threat summary. ``Addressed'' means Algorithms~\ref{alg:retrieve}--\ref{alg:write}
mitigate the threat directly under the stated assumptions; ``Partial'' means
the mechanism reduces but does not eliminate exposure; ``No'' means an
external control is required.}
\label{tab:threats}
\footnotesize
\begin{tabular}{@{}p{0.07\linewidth}p{0.44\linewidth}p{0.1\linewidth}p{0.29\linewidth}@{}}
\toprule
\textbf{ID} & \textbf{Threat} & \textbf{Status} & \textbf{Production control} \\
\midrule
T1 & Sensitive-column leakage via cached result & Addressed & Lineage gate (Alg.~\ref{alg:retrieve}) \\
T2 & Silent metric-definition drift & Addressed & Conflict detector (Alg.~\ref{alg:write}) \\
T3 & Stale sensitive data after policy update & Partial & TTL + lazy $\mathrm{S}(a)$ recompute (Cor.~\ref{cor:policy-update}) \\
T4 & Inference attacks (undeclared derived features) & No & Manual review of feature engineering; out of scope \\
T5 & Malicious/incomplete lineage fabrication & Partial & Automatic extraction (\textsc{SQLGlot}/LINEAGEX), \S\ref{sec:degradation} \\
T6 & Aggregate inference across retrievals & No & Differential privacy layer (future work) \\
T7 & Insider threats on the store & No & Trusted infrastructure, audit logging \\
T8 & Prompt injection / timing side-channel on lineage & Partial & Automatic extraction; constant-time gate (optional) \\
\bottomrule
\end{tabular}
\end{table*}

The threats the mechanism addresses (T1, T2) are exactly the two failure
modes that motivated this work. The partially addressed threats (T3, T5,
T8) share a common root cause: they all come down to relying on lineage
completeness and integrity, which is why automatic, tamper-resistant
lineage extraction helps with all three. The unaddressed threats (T4, T6,
T7) mark where a memory-layer control simply runs out of reach -- they
need cooperation from the feature-engineering layer, a
differential-privacy layer, or the infrastructure trust boundary that AMU
does not control.
In short, AMU is one layer in a defence-in-depth stack, and we do not
claim it is a complete security solution by itself.
Two further attack vectors fall outside T1--T8 and are worth naming
explicitly rather than leaving implicit. \emph{Cross-metric correlation}:
a requester who is individually denied column $c$ may still triangulate
it by combining several \emph{permitted} AMUs whose values are jointly
informative about $c$, since the gate reasons about one retrieval at a
time, not about what a sequence of safe retrievals reveals in
combination -- a multi-request generalisation of T6. \emph{Departmental
collusion}: two departments each authorised for a disjoint subset of a
sensitive derivation could share their individually safe cached results
out-of-band to reconstruct the full picture, which no memory-layer
control can prevent once results leave the store. Both are, like T6 and
T7, beyond what a single-retrieval column gate can address by
construction, and we flag them as motivation for pairing AMU with
usage-auditing and differential-privacy controls at the organisational
level.

\section{Experiments}
\label{sec:experiments}

We built a synthetic discrete-event simulator in pure Python that isolates
the retrieval-and-conflict-detection mechanism from LLM inference, network
latency, and query-optimiser confounds. Because every event's ground-truth
lineage, sensitivity, and department permissions are known exactly, we can
measure leak and reuse rates without label noise.
We compare three systems: \textbf{No Memory} (leak-free, zero reuse),
\textbf{Naive Shared Memory} (content-gated, representing prior systems),
and \textbf{Lineage-Aware} (our proposal).
All experiments report mean $\pm$ population standard deviation over 30
seeds, with bootstrap 95\% confidence intervals and Mann-Whitney $U$
significance tests; the full protocol and raw per-seed tables are in the
Supplementary Material.

\subsection{Experiments 1--2: Synthetic and TPC-H Schemas}
\label{sec:exp1}
\label{sec:exp2}

Experiment~1 uses a 5-table synthetic schema (three departments, sensitive
columns $\mathcal{S}=\{\texttt{ssn},\texttt{email},\texttt{income}\}$, Table~\ref{tab:perms});
Experiment~2 replicates it on the 8-table TPC-H decision-support
benchmark~\cite{tpch2021} ($|\mathcal{S}_{\text{TPC-H}}|=9$) to test whether
findings generalise to a denser, more realistic schema.

\begin{table}[t]
\centering
\caption{Department permissions, synthetic schema. The gate enforces
restrictions only on columns in $\mathcal{S}$; non-sensitive blocked columns
represent full policy intent but are not currently gate-enforced.}
\label{tab:perms}
\footnotesize
\begin{tabular}{@{}p{0.16\linewidth}p{0.34\linewidth}p{0.38\linewidth}@{}}
\toprule
\textbf{Dept.} & \textbf{Permitted (summary)} & \textbf{Blocked (in $\mathcal{S}$)} \\
\midrule
Finance   & All columns            & None \\
Marketing & Commercial, order data & \texttt{ssn, income, email} \\
Support   & Customer, ticket meta  & \texttt{ssn, income, email} \\
\bottomrule
\end{tabular}
\end{table}

\begin{table}[t]
\centering
\caption{Leak and reuse rates, both schemas (30 seeds, mean $\pm$ std dev).
All naive--LA differences: bootstrap 95\% CI excludes zero, Mann-Whitney
$U$, $p<0.001$. Full tables in the Supplementary Material.}
\label{tab:combined-results}
\footnotesize
\begin{tabular}{@{}p{0.24\linewidth}p{0.15\linewidth}p{0.15\linewidth}p{0.15\linewidth}p{0.15\linewidth}@{}}
\toprule
& \multicolumn{2}{c}{\textbf{Synthetic (5-tab.)}} & \multicolumn{2}{c}{\textbf{TPC-H (8-tab.)}} \\
\textbf{Metric} & \textbf{Naive} & \textbf{LA} & \textbf{Naive} & \textbf{LA} \\
\midrule
Leak rate (\%) & $18.8\pm3.7$ & $0.0\pm0.0^\dagger$ & $25.5\pm4.5$ & $0.0\pm0.0^\dagger$ \\
Reuse rate (\%) & $95.8\pm0.0$ & $82.6\pm2.3$ & $95.8\pm0.0$ & $81.5\pm4.8$ \\
Conflict recall (\%) & $0.0\pm0.0$ & $100.0\pm0.0$ & $0.0\pm0.0$ & $100.0\pm0.0$ \\
\bottomrule
\multicolumn{5}{l}{\footnotesize$^\dagger$Formal guarantee (Theorem~\ref{thm:safety}), not an empirical discovery.}
\end{tabular}
\end{table}

\begin{figure}[t]
\centering
\includegraphics[width=\linewidth]{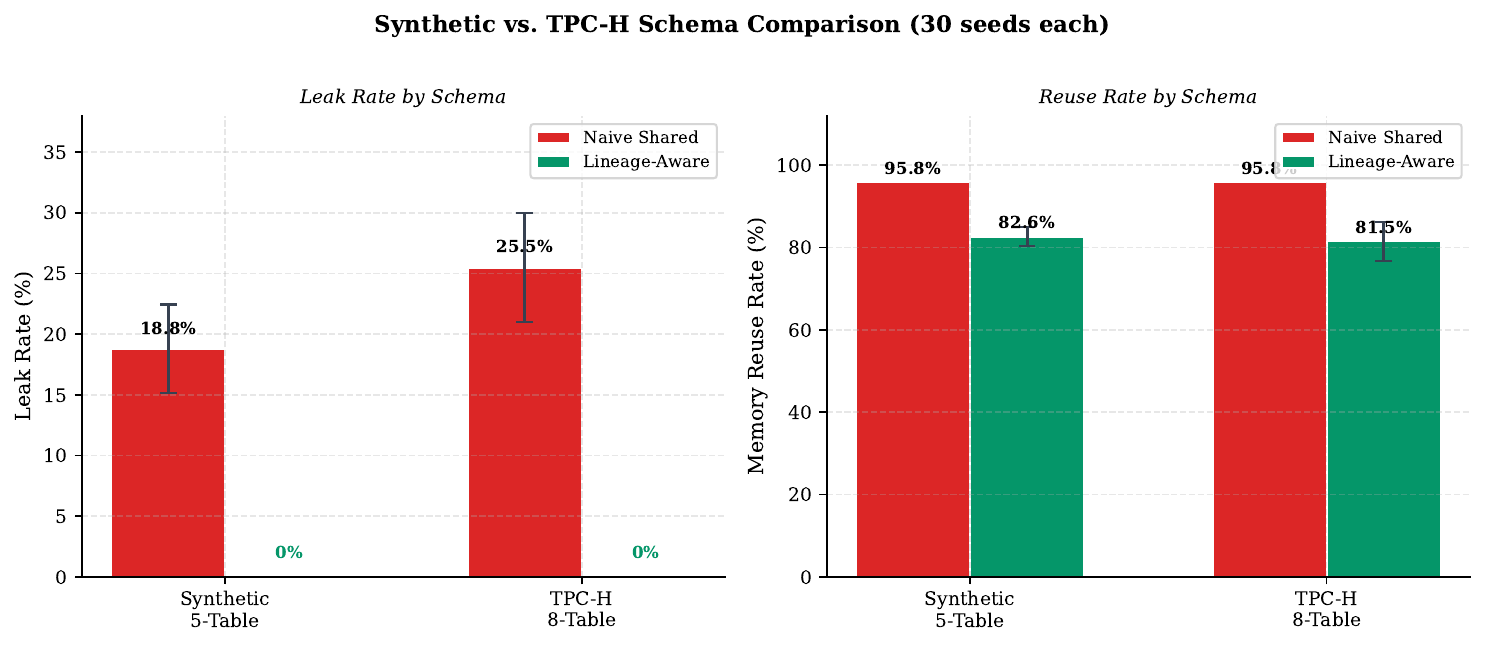}
\caption{Synthetic vs.\ TPC-H schema comparison (30 seeds each). TPC-H's
denser sensitive-column structure yields a higher naive leak rate; lineage-aware
gating eliminates leakage in both configurations.}
\label{fig:tpch}
\end{figure}

To be clear about what these numbers mean: the lineage-aware system's
$0\%$ leak rate is a \emph{policy guarantee} from
Theorem~\ref{thm:safety}, not something we discovered empirically. The
actual empirical finding is that the naive baseline's leak rate is
structurally non-zero (effect size $d \approx 7.2$--$8.0$, well past the
conventional ``large'' threshold of $0.8$), which reflects real retrievals
of sensitive-derivation results by departments that could not have run
the underlying queries themselves.
That protection is not free: the governance cost is a $13$--$14$
percentage-point reuse reduction, i.e., roughly 13 additional query
executions per 100 requests. That cost shows up as additional compute
rather than governance latency, since the gate-check predicate itself
costs only $0.09\,\mu$s (\S\ref{sec:complexity}).
This tradeoff held up consistently across two substantially different
schemas (60 total seed runs), which suggests content-gated leakage is a
structural property of the failure mode -- any workload that mixes
sensitive and non-sensitive derivation paths under shared metric names
will show it, rather than this being an artefact of how we built the
workload.
The two schemas differ in sensitive-registry size ($|\mathcal{S}|{=}3$
vs.\ $9$) and join density. If we treat this as a cross-experiment
ablation rather than two independent data points, holding the
sensitive-variant proportion constant (2 of 5 metrics in both schemas)
while tripling $|\mathcal{S}|$ and increasing join density lines up with
the leak-rate increase from $18.8\%$ to $25.5\%$. That points to leak rate
scaling with \emph{how many ways} a query can incidentally touch a
sensitive column, rather than with how many metrics are nominally
sensitive -- though this is consistent with a monotonic relationship, not
proof of one. A dedicated $2{\times}2{\times}2$ factorial sweep over step
count, registry size, and sensitive-variant proportion (specified in the
Supplementary Material) would settle the question directly, and we list
it as the top empirical priority for follow-on work.

\subsection{Experiment 3: Lineage-Completeness Degradation}
\label{sec:degradation}

Assumption~\ref{ass:complete} is the load-bearing condition of
Theorem~\ref{thm:safety}; in practice agents may under-report lineage due to
bugs or opaque query plans.
We define completeness $\rho \in [0,1]$ as the fraction of each step's
columns correctly reported, and measure the \emph{true} leak rate via a
shadow store that records full ground-truth sensitivity alongside the
truncated lineage the system actually sees.

\begin{figure}[t]
\centering
\includegraphics[width=\linewidth]{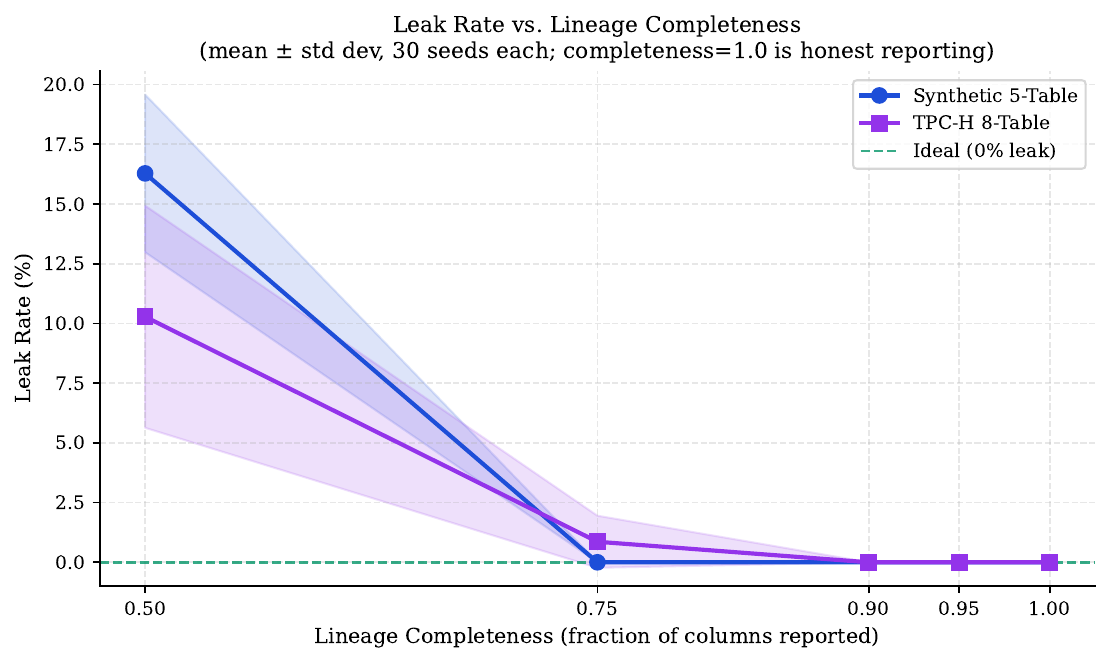}
\caption{Leak rate vs.\ lineage completeness $\rho$ (30 seeds, mean
$\pm$1\,std dev shading). Both schemas converge to 0\% by $\rho{=}0.75$--$0.90$;
at $\rho{=}0.50$ the synthetic schema approaches the naive baseline.}
\label{fig:degradation}
\end{figure}

At $\rho{=}0.50$, truncated lineage gives almost no protection (synthetic:
$16.3\pm3.3\%$ leak; TPC-H: $10.3\pm4.7\%$). At $\rho\geq0.75$ (synthetic)
or $\rho\geq0.90$ (TPC-H), the system reaches $0\%$ measured leak rate and
converges on the formal guarantee.
TPC-H needs higher completeness because its sensitive columns are spread
across more join steps -- join depth acts as a sensitivity multiplier on
the completeness requirement, not just on raw leak rate.
Our recommendation is $\rho \geq 0.9$ as a minimum target, and we would
lean toward automatic lineage extraction (Experiment~6) over agent
self-reporting, since self-reporting is hard to audit and degrades
unpredictably under exactly the conditions -- complex multi-hop joins --
where it matters most.

\subsection{Experiment 4: Extended Fuzzy Conflict-Detection Study}
\label{sec:fuzzy}

The exact-hash detector (Algorithm~\ref{alg:write}) flags any definition-hash
mismatch as a conflict.
We extended an original 10-pair evaluation to 43 hand-labelled lineage pairs
across four categories---15 True Conflicts, 15 Threshold Variants, 8
Logic-Operator changes (AND$\to$OR), and 5 Column-Subset cases---to test
three detectors: D1 (exact hash), D2 (Jaccard filter, $\tau{=}0.70$), and D3
(column-graph).

\begin{table}[t]
\centering
\caption{Extended fuzzy conflict-detection results (43 pairs). Full
per-category breakdown in the Supplementary Material.}
\label{tab:fuzzy}
\footnotesize
\begin{tabular}{@{}lcccc@{}}
\toprule
\textbf{Detector} & \textbf{Prec.} & \textbf{Recall} & \textbf{F\textsubscript{1}} & \textbf{LO correct} \\
\midrule
D1 --- Exact hash          & 0.651 & 1.000 & 0.789 & 8/8 \\
D2 --- Jaccard ($\tau{=}0.70$) & 0.676 & 0.821 & 0.742 & 3/8 \\
D3 --- Column-graph        & \textbf{1.000} & 0.714 & 0.833 & 0/8 \\
\bottomrule
\end{tabular}
\end{table}

What we found is a blind spot nobody had reported before: D3, which
achieved perfect $F_1{=}1.00$ on the original 10-pair set, misses all 8
Logic-Operator pairs once the dataset grows ($F_1$ drops to $0.833$ on 43
pairs), because it ignores filter-string variation entirely. That is the
right call for threshold recalibrations, but it leaves D3 blind to
AND$\to$OR changes that alter metric semantics without touching the
column set.
D1 catches all 8 LO pairs but pays for it with weak precision ($0.651$,
false-alerting on every threshold recalibration).
Neither detector is production-ready on its own; we propose, though we
have not yet evaluated, a composite D3+filter-logic-structure detector D4,
specified in the Supplementary Material.
The broader lesson here is that fuzzy-detection evaluations on small,
narrow datasets can overstate performance simply by missing edge cases.

\begin{figure}[t]
\centering
\includegraphics[width=\linewidth]{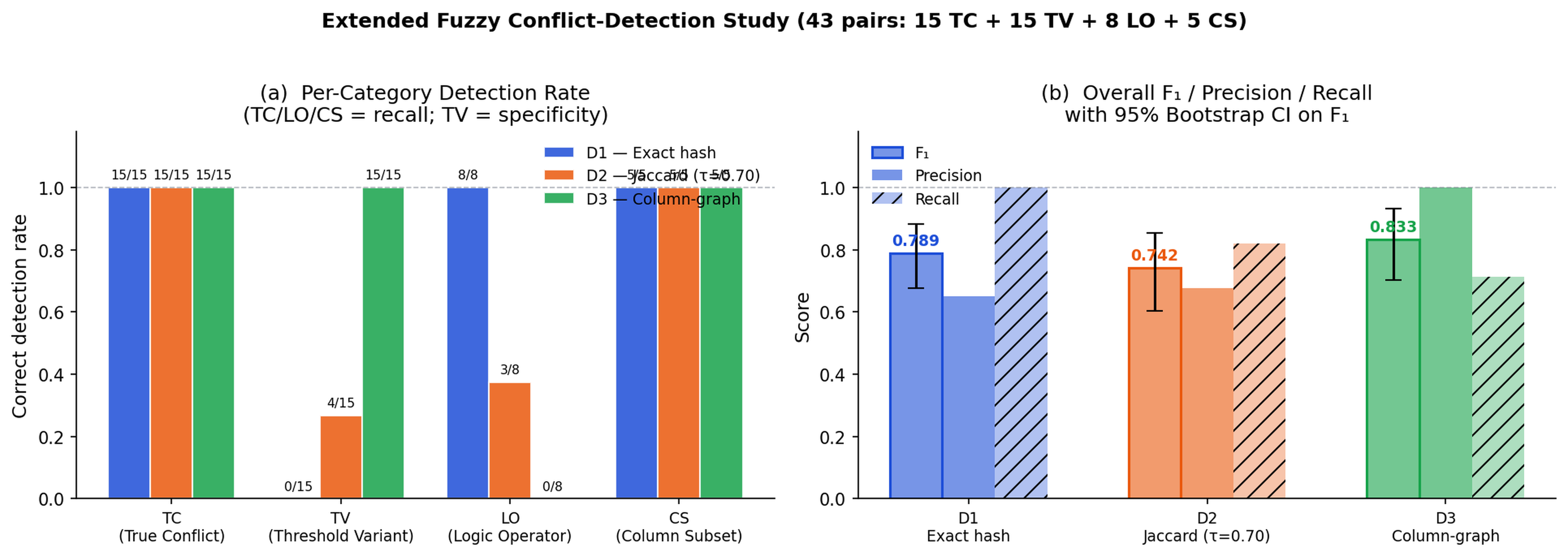}
\caption{Extended fuzzy conflict-detection results. \textbf{(a)}
Per-category detection rate: recall for TC/LO/CS, specificity for TV. D3's
perfect specificity on TV collapses to zero detection on LO. \textbf{(b)}
Overall $F_1$/precision/recall with 95\% bootstrap CI on $F_1$.}
\label{fig:fuzzy_ext}
\end{figure}

\subsection{Experiment 5: Empirical Runtime Analysis}
\label{sec:runtime}

We instrument all core operations with \texttt{time.perf\_counter()} over
2,000 iterations at five store sizes ($n\in\{1,5,10,20,50\}$), single-threaded
in a standard containerised Python environment with no specialised
acceleration.
\textbf{Hardware and environment:} all measurements were collected on a
4-core Apple Silicon (ARM64) virtual machine with 3.8\,GB RAM, running
Ubuntu~22.04\,LTS (Linux kernel 6.8) and CPython~3.10.12, inside a Docker
container with no GPU, SIMD intrinsics, or JIT acceleration enabled; no other
CPU-bound process ran concurrently during measurement.
We report this configuration in full because the absolute microsecond
values are platform-dependent. The safety-overhead argument in
\S\ref{sec:complexity} only relies on the $O(n)$-vs-$O(1)$ \emph{scaling}
behaviour, which is a property of the algorithm rather than of this
particular machine -- but we have not independently re-measured on a
second hardware platform yet, so we list cross-platform confirmation as a
reproducibility item (\S\ref{sec:futurework}).
Table~\ref{tab:complexity} already reports the key numbers: retrieve
worst-case scales linearly from $0.66\,\mu$s ($n{=}1$) to $13.82\,\mu$s
($n{=}50$), a $21\times$ increase confirming $O(n)$; best-case retrieve and
early-exit conflict detection stay flat, confirming $O(1)$.
Even the worst-case overhead is 2--3 orders of magnitude below typical
production analytics-query latency (single-digit ms to seconds), confirming
the mechanism's practical deployability.
The complete runtime sweep across all five store sizes, including
per-iteration variance, is tabulated in the Supplementary Material.

\subsection{Experiment 6: Real-Agent Integration}
\label{sec:exp6}

Experiments~1--5 assume complete lineage (Assumption~\ref{ass:complete}).
This experiment tests one practical route toward satisfying that
assumption: we replace agent self-declared lineage with \emph{automatic}
extraction via
\textsc{sqlglot}~\cite{sqlglot2023}, a pure-Python SQL parser, interposed
between the agent and the database cursor---every table and column in the
executed query's abstract syntax tree (AST) is extracted deterministically, so the agent cannot
under-report because it never declares anything.
We use the Northwind database~\cite{northwind2000} with four departments
(Finance, Sales, Operations, HR) and sensitive registry
$\{\texttt{unitprice},\texttt{freight},\texttt{homephone},\texttt{birthdate}\}$,
driven by a LangChain~\cite{langchain2023}+Ollama~\cite{ollama2024} agent
issuing real (and demo-mode canned) SQL across three scenarios covering the
gate's three outcomes: safe reuse, blocked-with-conflict, and vacuous reuse
on a metric with no sensitive columns.

Across all 9 cross-department round-trips, zero sensitive-column leaks
occur: 5 fresh SQL queries are executed and written (2 following a gate
block), 4 requests are served entirely from cache, and 2 gate blocks
precisely coincide with the 2 automatically detected metric-definition
conflicts---i.e., every block corresponded to a genuine semantic divergence,
not a false positive.
Cache-hit latency ($0.009\,$ms) comes in three to four orders of magnitude
below the SQL it replaces, which lines up with the isolated benchmarks of
\S\ref{sec:runtime}.
This is a nine-round-trip proof-of-concept, not a large-scale trial: it
shows automatic SQL lineage extraction can satisfy
Assumption~\ref{ass:complete} for the SQL query patterns exercised here,
with no change to agent logic and no LLM cooperation required. Whether it
extends to broader SQL workloads is still an open question, particularly given the
view-, stored-procedure-, and aliasing-related caveats we discuss in
\S\ref{sec:discussion}, and the larger-scale validation we list as future
work (\S\ref{sec:futurework}).
Full scenario descriptions and the round-trip-by-round-trip table are in
the Supplementary Material.

\section{Discussion and Conclusion}
\label{sec:discussion}

\subsection{Enterprise Value and What the Experiments Establish}

\S\ref{sec:experiments} points to three benefits for enterprise
deployments. On \emph{compliance posture}, Theorem~\ref{thm:safety} turns
a policy assertion into something mechanically checkable -- an auditor can
verify it by inspecting gate logic instead of trusting agent behaviour,
which helps with EU AI Act explainability obligations~\cite{euaiact2024}.
On \emph{avoided breach exposure}, eliminating an $18.8$--$25.5\%$
structural leak channel changes the risk profile materially given a
\$4.44M average breach cost~\cite{ibm2025breach}. And on \emph{bounded
efficiency cost}, the $13$--$14$ point reuse reduction is a one-time,
quantified governance tax rather than an open-ended latency regression,
since the gate itself costs sub-microsecond time.
Stepping back, the six experiments together show that content-gated
memory systematically leaks across departments (Experiments 1--2); that
the $0\%$ lineage-aware leak rate follows from the theorem, so the real
empirical content is the reuse-cost tradeoff; that the safety guarantee
degrades gracefully rather than catastrophically below full completeness
(Experiment~3); that production conflict detection needs a composite
detector (Experiment~4); that gating overhead is negligible
(Experiment~5); and that automatic SQL lineage extraction is a practical
mechanism for improving lineage completeness in a real agent stack, with
zero observed leaks across the 9 round-trips exercised (Experiment~6).

\subsection{Limitations and Future Work}

\textbf{Limitations:} Several of these are worth stating plainly. The
$18.8$--$25.5\%$ naive leak rates are properties of our workload mixes, not
universal base rates, though the structural argument holds regardless of
the specific percentage. The mechanism does not catch inference attacks
that encode sensitive information in a derived feature without naming the
source column (T4, the most important gap for production use). The
43-pair fuzzy dataset is still small. And the runtime benchmarks
(\S\ref{sec:runtime}) were measured on a single hardware/OS configuration,
so the reported constants -- though not the $O(n)$-vs-$O(1)$ scaling --
may shift on other platforms.

\textbf{View, stored-procedure, and alias limitations of automatic lineage
extraction.} Experiment 6's automatic-extraction guarantee -- closing
Assumption~\ref{ass:complete} without needing agent cooperation -- is only
as complete as \textsc{SQLGlot}'s static AST parse of the executed SQL
text, and that leaves three gaps worth flagging.
First, when a query selects from a database \emph{view} rather than a base
table, \textsc{SQLGlot} records the view name as the lineage source unless
it is also given the view's defining SQL to expand. If the view definition
does not get resolved (say it lives only in the database catalog and was
never passed to the extractor), a sensitive base column hidden behind an
innocuously-named view (\texttt{v\_customer\_summary}) will not show up in
$L$, which quietly violates Assumption~\ref{ass:complete}.
Second, stored procedures and user-defined functions invoked from SQL
(e.g., \texttt{CALL compute\_churn(...)}) are opaque to AST-level column
extraction -- the procedure body might touch sensitive columns internally
that never appear as identifiers in the calling statement, so capturing
lineage here means either inlining the procedure body before parsing or
maintaining a separate column-level lineage annotation for each registered
procedure.
Third, column and table \emph{aliasing} (\texttt{SELECT c.income AS c1},
or self-joins that alias the same physical table twice) is usually
resolved correctly by \textsc{SQLGlot}'s scope analysis, but deeply nested
aliasing across several levels of subquery or common table expression
(CTE) can defeat that resolution in practice, especially combined with
\texttt{SELECT~*} expansion against a catalog the extractor was not given.
In all three cases, the system's fail-closed design
(\S\ref{sec:implementation}) limits the damage to over-blocking -- an AMU
that should have been retrievable gets marked maximally sensitive --
rather than actual leakage. Still, this is a completeness gap in
Experiment 6's zero-leak result that a purely syntactic extractor cannot
close on its own; closing it needs catalog-aware view/procedure expansion,
which we flag as future work rather than something the current
implementation already does.

\textbf{Future work, in priority order:}\label{sec:futurework}
(1)~adversarial lineage testing---red-team
queries that smuggle sensitive joins through derived features or
view indirection, stress-testing Assumption~\ref{ass:complete}; (2)~non-SQL
lineage extraction for vector-store, REST/GraphQL, and in-process
derivations, and catalog-aware view/stored-procedure expansion for the SQL
case; (3)~a production deployment study on a real 20--50 table
warehouse over weeks rather than seeds, on more than one hardware platform;
and (4)~implementing and evaluating
the proposed D4 composite conflict detector against a larger real-analytics
corpus, alongside open-sourcing the reference implementation.

\subsection{Conclusion}

Every current AI agent memory system governs \emph{what} is stored and
\emph{who} may access it, but none of them track \emph{how} the stored
result was derived. Restricted to that scope -- named, logged columns
under reasonably complete lineage, not the T4 inference case discussed
above -- the Analytical Memory Unit schema and lineage-gated retrieval
policy close that gap at a modest overhead: $4$--$8\times$
storage, a sub-microsecond gate check, and at most $13.8\,\mu$s
worst-case retrieval latency -- backed by a formal safety guarantee
(Theorem~\ref{thm:safety}) and six experiments that characterise its
behaviour, its degradation under incomplete lineage, and its preliminary
feasibility on a small real-agent integration -- not yet a demonstration
of viability at production scale.
As AI agents move from single-purpose assistants toward being
interconnected participants in enterprise data workflows, we expect the
systems that govern what those agents remember to matter more and more
for what those agents end up being allowed to say, and to whom.
Content-and-role governance answers ``who owns this memory''; what we
think is the equally necessary question is one that derivation-gated
governance can answer instead: could this requester have produced this
result themselves?

\section{Reproducibility, Data, and Ethics}
\label{sec:repro}

Our reference implementation -- Python code covering the AMU/lineage data
model, three competing memory systems, an event-replay simulator,
statistical analysis, and the \textsc{sqlglot}-based real-agent pipeline
-- reproduces every table and figure in this paper, and we containerised
it via Docker Compose for one-command reproduction. Full
module-by-module documentation is in the Supplementary Material.
We did not use any proprietary or personally identifiable data: the
synthetic and TPC-H-inspired schemas use generated or benchmark data, and
the Northwind database (Experiment~6) is a decades-old public fictional
dataset with no real individuals' information in it.
This work does not involve human subjects or data requiring ethical
review. We were deliberate about making the threat model
(\S\ref{sec:threat}) explicit about what the mechanism does \emph{not}
protect against, since overstating a safety guarantee's scope is itself a
governance risk.
We will make source code, generated datasets, and raw per-seed results
available at a public repository upon acceptance; in the interim they are
available from the corresponding author on reasonable request.

\section*{Conflict of Interest}
The authors declare that they have no known competing financial interests
or personal relationships that could have appeared to influence the work
reported in this paper.

\section*{Acknowledgment}
The authors thank no additional contributors beyond those listed in the
author list.

\balance

\vspace{6pt}

\begin{IEEEbiographynophoto}{Venkata Sangaraju}
is an independent researcher working on memory and data-governance
architectures for enterprise artificial intelligence (AI) agents. His
research interests include lineage and provenance tracking for AI agent
systems, privacy-preserving retrieval, column-level access control, and the
application of formal safety guarantees to shared multi-agent memory
infrastructure. He is the corresponding author of this work
(ORCID: 0009-0001-7716-1342).
\end{IEEEbiographynophoto}

\begin{IEEEbiographynophoto}{Sudhir Vissa}
is with SAGE7 AI, Georgetown, Texas, USA. His research interests include
enterprise AI agent architectures, data governance, and applied
infrastructure for multi-agent deployments, including the
metric-definition and access-control problems addressed in this work
(ORCID: 0009-0003-7865-0863).
\end{IEEEbiographynophoto}
\EOD

\end{document}